\documentclass[a4paper, 11pt]{article}
\usepackage[utf8]{inputenc}
\usepackage[ngermanb, english]{babel}
\usepackage{amsmath}
\usepackage{amssymb}
\usepackage{amsthm}
\usepackage{euscript}
\usepackage{tikz-cd}
\usepackage{bm}
\usepackage{commath}
\usepackage{mathrsfs}
\usepackage{mathtools}
\usepackage{slashed}
\usepackage{tensor}
\usepackage{parskip}
\usepackage{epstopdf}
\usepackage{graphicx}
\usepackage{hhline}
\usepackage[nottoc]{tocbibind}
\usepackage{babelbib}
\usepackage{hyperref}
\usepackage[top=2.5cm, left=2.5cm, right=2.5cm, bottom=2.5cm]{geometry}

\theoremstyle{plain}
\newtheorem{thm}{Theorem}[section]

\newtheorem{col}[thm]{Corollary}

\theoremstyle{definition}

\newtheorem{con}[thm]{Convention}

\newtheorem{exmp}[thm]{Example}

\theoremstyle{remark}
\newtheorem{rem}[thm]{Remark}

\newcommand{\deDonder}{{d \negmedspace D \mspace{-2mu}}}

\newcommand{\ordersb}[1]{O ( #1 )}

\newcommand{\dt}[1]{\operatorname{Det} \left ( #1 \right )}

\newcommand{\textfrac}[2]{#1 / #2}
\newcommand{\imaginary}{\mathrm{i}}

\newcommand{\enter}{\vspace{\baselineskip}}

\newcommand{\mathbbit}[1]{{\mspace{-1mu} \italicbox{$\mathbb{#1}$} \mspace{2mu}}}

\newcommand{\bbsL}{\mathbbit{\scriptstyle{L}}}

\newcommand{\commutatorbig}[2]{\big [ #1 , #2 \big ]}
\newcommand{\setbig}[1]{\big \{ #1 \big \}}
\newcommand{\gcoupling}{\varkappa}
\newcommand{\gravfr}{\mathfrak{G}}
\newcommand{\pregravfr}{\mathfrak{g}}
\newcommand{\ccgravfr}{\mathfrak{L}}
\newcommand{\preccgravfr}{\mathfrak{l}}

\newcommand{\gravprop}{\mathfrak{P}}
\newcommand{\gravghostprop}{\mathfrak{p}}

\newcommand{\triplevert}{\vert\kern-0.25ex\vert\kern-0.25ex\vert}
\newcommand{\cgreen}[1]{\vcenter{\hbox{\includegraphics[width=\cgreenlength]{#1}}}}
\newcommand{\tcgreen}[1]{\vcenter{\hbox{\includegraphics[width=0.75\cgreenlength]{#1}}}}
\newcommand{\bgm}{b}

\newsavebox{\foobox}
\newcommand{\italicbox}[2][.25]
{%
	\mbox
	{%
		\sbox{\foobox}{#2}%
		\hskip\wd\foobox
		\pdfsave
		\pdfsetmatrix{1 0 #1 1}%
		\llap{\usebox{\foobox}}%
		\pdfrestore
	}%
}

\makeatletter
\newcommand{\subalign}[1]{%
  \vcenter{%
    \Let@ \restore@math@cr \default@tag
    \baselineskip\fontdimen10 \scriptfont\tw@
    \advance\baselineskip\fontdimen12 \scriptfont\tw@
    \lineskip\thr@@\fontdimen8 \scriptfont\thr@@
    \lineskiplimit\lineskip
    \ialign{\hfil$\m@th\scriptstyle##$&$\m@th\scriptstyle{}##$\crcr
      #1\crcr
    }%
  }
}
\makeatother

\providecommand{\eqnref}[1]{Equation~\eqref{#1}}
\providecommand{\eqnsref}[1]{Equations~\eqref{#1}}
\providecommand{\eqnsaref}[2]{Equations~\eqref{#1} and \eqref{#2}}

\providecommand{\colssaref}[3]{Corollaries~\ref{#1}, \ref{#2} and \ref{#3}}

\DeclareSymbolFont{extraitalic}      {U}{zavm}{m}{it}
\DeclareMathSymbol{\Qoppa}{\mathord}{extraitalic}{161}
\DeclareMathSymbol{\qoppa}{\mathord}{extraitalic}{162}
\DeclareMathSymbol{\Stigma}{\mathord}{extraitalic}{167}
\DeclareMathSymbol{\Sampi}{\mathord}{extraitalic}{165}
\DeclareMathSymbol{\sampi}{\mathord}{extraitalic}{166}
\DeclareMathSymbol{\stigma}{\mathord}{extraitalic}{168}

\newlength{\graphlength}
\newlength{\cgreenlength}
\title{\textsc{On Perturbative Quantum Gravity with a Cosmological Constant}}
\author{David Prinz\footnote{Riemann fellow at Riemann Center for Geometry and Physics and Institute for Theoretical Physics at Leibniz University Hannover; prinz@itp.uni-hannover.de}}
\date{March 24, 2023}

\begin{document}

\maketitle

\begin{abstract}
	We discuss how the incorporation of a cosmological constant affects the perturbative quantization of (effective) Quantum General Relativity. To this end, we derive the gravitational Slavnov--Taylor identities and appropriate renormalization conditions for the cosmological constant. Additionally, we calculate the corresponding Feynman rules for any vertex valence and with general gauge parameter. Furthermore, we provide the BRST setup and generate the Faddeev--Popov ghost and the symmetric ghost via a gauge fixing fermion and a gauge fixing boson, respectively. Finally, we study the transversality of the graviton propagator and the graviton three-valent vertex.
\end{abstract}

\section{Introduction} \label{sec:introduction}

Perturbative Quantum Gravity is generally believed not to admit a predictive perturbative expansion. This is due to the fact that it is non-renormalizable by power-counting and thus requires higher-derivative counterterms. These have been examined on the level of one and two loops \cite{tHooft_Veltman_OLD,Goroff_Sagnotti_TLD} in the framework of the background field method \cite{Abbott_BFM}. Surprisingly, the counterterms on one-loop level can be combined, using the Chern--Gauss--Bonnet theorem, to a total derivative \cite{tHooft_Veltman_OLD}. This means that the divergences can be absorbed via a field redefinition. On the other hand, on two-loop level such a relation is missing, leading to the well-accepted conjecture that perturbative Quantum Gravity is not renormalizable at two loops \cite{Goroff_Sagnotti_TLD}. Notably, it is still possible to find a non-local field redefinition to absorb the divergent contributions \cite{Krasnov}. We reapproach this long-standing renormalization problem from a different viewpoint, namely via the framework of Hopf algebraic renormalization \cite{Kreimer_Hopf_Algebra,Connes_Kreimer_NG,Connes_Kreimer_0,Connes_Kreimer_1,Connes_Kreimer_2}: Soon after the development of this formalism, it has been applied to gauge theories \cite{Kreimer_Anatomy,vSuijlekom_QED,vSuijlekom_QCD,vSuijlekom_BV,Kreimer_vSuijlekom}. In particular, D.\ Kreimer applied the same reasoning to perturbative Quantum Gravity and conjectured that its non-renormalizable nature could effectively become harmless, if it satisfies a tower of infinite gravitational Slavnov--Taylor identities \cite{Kreimer_QG1,Kreimer_QG2}. This has been worked out in detail by the author in a series of recent articles \cite{Prinz_2,Prinz_3,Prinz_4,Prinz_5,Prinz_6,Prinz_7} and his dissertation \cite{Prinz_PhD}. In particular, therein the Connes--Kreimer renormalization theory has been generalized to non-renormalizable (generalized) gauge theories, so that it applies to (effective) Quantum General Relativity. Specifically, we have derived precise identities which are necessary for a well-defined perturbative expansion \cite{Prinz_3}. These identities need to be checked on the level of the gravitational Feynman rules given in \cite{Prinz_4}, which is a topic of ongoing research. In the present article, we want to address the case of a non-vanishing cosmological constant \(\Lambda\). In particular, we calculate the graviton vertex Feynman rules for arbitrary vertex valence in \thmref{thm:grav-fr-cc} and the graviton propagator for the linearized de Donder gauge fixing with general gauge parameter \(\zeta\) in \thmref{thm:grav-prop-cc}. In addition, we list the three-valent and four-valent graviton vertex Feynman rules explicitly in \exref{exmp:grav-fr-cc}.

To this end, we use the following Lagrange density for (effective) Quantum General Relativity with a cosmological constant (QGR-\(\Lambda\)):
\begin{subequations} \label{eqns:lagrange-density-qgrcc}
\begin{align}
	\mathcal{L}_\text{QGR-\(\Lambda\)} & := \mathcal{L}_\text{GR-\(\Lambda\)} + \mathcal{L}_\text{GR-\(\Lambda\)-GF} + \mathcal{L}_\text{GR-\(\Lambda\)-Ghost}
	\intertext{with the Einstein--Hilbert Lagrange density with a cosmological constant}
	\mathcal{L}_\text{GR-\(\Lambda\)} & := - \frac{1}{2 \varkappa^2} \left ( R + 2 \Lambda  \right ) \dif V_g \, ,
	\intertext{the linearized de Donder gauge fixing Lagrange density}
	\mathcal{L}_\text{GR-\(\Lambda\)-GF} & := - \frac{1}{4 \varkappa^2 \zeta} \bgm^{\mu \nu} \deDonder^{(1)}_\mu \deDonder^{(1)}_\nu
	\intertext{and either the Faddeev--Popov ghost construction, cf.\ \colref{col:linearized_de_donder_gauge_fixing_fermion-cc},}
	\mathcal{L}_\text{GR-\(\Lambda\)-FP-Ghost} & := - \frac{1}{2} \bgm^{\rho \sigma} \left ( \frac{1}{\zeta} \overline{C}^\mu \left ( \partial_\rho \partial_\sigma C_\mu \right ) + \overline{C}^\mu \left ( \partial_\mu \big ( \tensor{\Gamma}{^\nu _\rho _\sigma} C_\nu \big ) - 2 \partial_\rho \big ( \tensor{\Gamma}{^\nu _\mu _\sigma} C_\nu \big ) \right ) \right ) \dif V_\bgm
	\intertext{or the symmetric ghost construction, cf.\ \colref{col:symmetric-ghost-qgrcc},}
	\begin{split}
	\mathcal{L}_\text{GR-\(\Lambda\)-Sym-Ghost} & := - \frac{1}{2 \zeta} \bgm^{\mu \nu} \big ( \partial_\mu \overline{C}^\rho \big ) \big ( \partial_\nu C_\rho \big ) \dif V_\bgm \\
		& \phantom{:=} - \frac{1}{2} \bgm^{\mu \nu} \overline{C}^\rho \left ( \frac{1}{2} \partial_\rho \big ( \tensor{\Gamma}{^\sigma _\mu _\nu} C_\sigma \big ) - \partial_\mu \big ( \tensor{\Gamma}{^\sigma _\rho _\nu} C_\sigma \big ) \right ) \dif V_\bgm \\
		& \phantom{:=} + \frac{1}{2} \bgm^{\mu \nu} \left ( \frac{1}{2} \partial_\rho \big ( \tensor{\Gamma}{^\sigma _\mu _\nu} \overline{C}_\sigma \big ) - \partial_\mu \big ( \tensor{\Gamma}{^\sigma _\rho _\nu} \overline{C}_\sigma \big ) \right ) C^\rho \dif V_\bgm \\
		& \phantom{:=} - \frac{\varkappa^2 \zeta}{8} \bgm_{\mu \nu} \left ( \overline{C}^\rho \big ( \partial_\rho \overline{C}^\mu \big ) \right ) \left ( C^\sigma \big ( \partial_\sigma C^\nu \big ) \right ) \dif V_\bgm \, .
		\end{split}
\end{align}
\end{subequations}
We remark that we consider its linearization with respect to the metric decomposition\footnote{This implies that the indices of the graviton field \(h_{\mu \nu}\) and its inverse \(h^{\mu \nu}\) are raised and lowered via the background metric \(\bgm_{\mu \nu}\) and its inverse \(\bgm^{\mu \nu}\), respectively.} \(g_{\mu \nu} \equiv \bgm_{\mu \nu} + \varkappa h_{\mu \nu}\), where \(h_{\mu \nu}\) is the graviton field, \(\varkappa := \sqrt{\kappa}\) the graviton coupling constant (with \(\kappa := 8 \pi G\) the Einstein gravitational constant) and \(\bgm_{\mu \nu}\) a convenient background metric.\footnote{Typically given via a solution of the vacuum Einstein field equations, i.e.\ either the de Sitter or the anti-de Sitter metric.} In addition, \(R := g^{\nu \sigma} \tensor{R}{^\mu _\nu _\mu _\sigma}\) is the Ricci scalar (with \(\tensor{R}{^\rho _\sigma _\mu _\nu} := \partial_\mu \tensor{\Gamma}{^\rho _\nu _\sigma} - \partial_\nu \tensor{\Gamma}{^\rho _\mu _\sigma} + \tensor{\Gamma}{^\rho _\mu _\lambda} \tensor{\Gamma}{^\lambda _\nu _\sigma} - \tensor{\Gamma}{^\rho _\nu _\lambda} \tensor{\Gamma}{^\lambda _\mu _\sigma}\) the Riemann tensor) and \(\Lambda\) the cosmological constant. Furthermore, \(\dif V_g := \sqrt{- \dt{g}} \dif t \wedge \dif x \wedge \dif y \wedge \dif z\) denotes the Riemannian volume form and \(\dif V_\bgm := \sqrt{- \dt{\bgm}} \dif t \wedge \dif x \wedge \dif y \wedge \dif z\) the background metric volume form. In addition, \(\deDonder^{(1)}_\mu := \bgm^{\rho \sigma} \Gamma_{\mu \rho \sigma} \equiv 0\) denotes the linearized de Donder gauge fixing functional and \(\zeta\) the gauge fixing parameter. Finally, \(C_\mu\) and \(\overline{C}^\mu\) are the graviton-ghost and graviton-antighost, respectively. In addition, we establish the relation between the Faddeev--Popov ghost construction and the symmetric ghost construction in \colref{col:symmetric-ghost-homotopy-qgrcc}, cf.\ \cite{Prinz_5,Prinz_6}.

To study multiplicative renormalization, we decompose \(\mathcal{L}_\text{QGR-\(\Lambda\)}\) into its individual monomials. As has been discussed in \cite{Prinz_3}, the individual monomials can be addressed via the powers in the graviton coupling constant \(\varkappa\), the gauge fixing parameter \(\zeta\), the ghost field \(C\) and now the cosmological constant \(\Lambda\) as follows\footnote{We omit the term \(\mathcal{L}_\text{QGR-\(\Lambda\)}^{-1,0,0,0}\) as it is given by a total derivative.}
\begin{equation}
	\mathcal{L}_\text{QGR-\(\Lambda\)} \equiv \sum_{i = 0}^\infty \sum_{j = -1}^0 \sum_{k = 0}^1 \sum_{l = 0}^1 \mathcal{L}_\text{QGR-\(\Lambda\)}^{i,j,k,l} \, ,
\end{equation}
where we have set \(\mathcal{L}_\text{QGR-\(\Lambda\)}^{i,j,k,l} := \eval[1]{\left ( \mathcal{L}_\text{QGR-\(\Lambda\)} \right )}_{O (\varkappa^i \zeta^j C^k \Lambda^l)}\), cf.\ \cite[Section 1]{Prinz_3} and \cite[Section 3]{Prinz_4}. Again, to absorb the upcoming divergences in a multiplicative manner, we multiply each monomial with an individual function \(Z^r \! \left ( \varepsilon \right )\) in the regulator \(\varepsilon \in \mathbb{R}\):
\begin{equation}
	\mathcal{L}_\text{QGR-\(\Lambda\)}^\text{R} \left ( \varepsilon \right ) := \sum_{i = 0}^\infty \sum_{j = -1}^0 \sum_{k = 0}^1 \sum_{l = 0}^1 Z^{i,j,k,l} \! \left ( \varepsilon \right ) \mathcal{L}_\text{QGR-\(\Lambda\)}^{i,j,k,l} \, ,
\end{equation}
where the regulator \(\varepsilon\) is related to the energy scale through the choice of a regularization scheme. Then, the invariance of \(\mathcal{L}_\text{QGR-\(\Lambda\)}^\text{R} \left ( \varepsilon \right )\) under (residual) diffeomorphisms away from the reference point (where all \(Z\)-factors fulfill \(Z^r \! \left ( \varepsilon_0 \right ) = 1\)) depends on the following identities:\footnote{These identities are known in the literature as Ward--Takahashi identity in the realm of Quantum Electrodynamics and Slavnov--Taylor identities in the realm of Quantum Chromodynamics \cite{Ward,Takahashi,tHooft,Taylor,Slavnov}.}
\begin{subequations} \label{eqns:z-factor_identities_qgrcc}
	\begin{align}
	\frac{Z^{i,0,0,0} \! \left ( \varepsilon \right ) Z^{1,0,0,0} \! \left ( \varepsilon \right )}{Z^{0,0,0,/} \! \left ( \varepsilon \right )} & \equiv Z^{(i+1),0,0,0} \! \left ( \varepsilon \right ) \, , \\
	\frac{Z^{i,0,0,1} \! \left ( \varepsilon \right ) Z^{1,0,0,1} \! \left ( \varepsilon \right )}{Z^{0,0,0,/} \! \left ( \varepsilon \right )} & \equiv Z^{(i+1),0,0,1} \! \left ( \varepsilon \right )
	\intertext{and}
	\frac{Z^{i,0,0,0} \! \left ( \varepsilon \right )}{Z^{0,-1,0,/} \! \left ( \varepsilon \right )} & \equiv \frac{Z^{i,0,1,0} \! \left ( \varepsilon \right )}{Z^{0,-1,1,0} \! \left ( \varepsilon \right )} \, ,
	\end{align}
\end{subequations}
for all \(i \in \mathbb{N}_+\) and \(\varepsilon\) in the domain of the regularization scheme. Here, we have denoted the transversal and longitudinal graviton propagator \(Z\)-factors via \(Z^{0,0,0,/} \! \left ( \varepsilon \right )\) and \(Z^{0,-1,0,/} \! \left ( \varepsilon \right )\), respectively, as they cannot be distinguished by their degree in the cosmological constant since it behaves like an Eucledian mass term. Rather, we suggest the following renormalization conditions, where \(\boldsymbol{\gravprop}_{\mu \nu \rho \sigma} \left ( p^\sigma, \Lambda; \zeta; \epsilon \right )\) denotes the dressed graviton propagator in the convenient decomposition of \thmref{thm:feynman-rule-graviton-propagator-cc-lt-decomposition}:
\begin{equation}
	\boldsymbol{\gravprop}_{\mu \nu \rho \sigma} \left ( p^\sigma, \Lambda; \zeta; \epsilon \right ) = - \frac{2 \imaginary p^2}{p^2 + \Lambda + \imaginary \epsilon + \Sigma_\gravprop \big ( p^2, \Lambda \big )} \left ( \mathbbit{G}_{\mu \nu \rho \sigma} - \left ( \frac{1 - \zeta}{p^2 + \zeta \Lambda + \Sigma_\gravghostprop \big ( p^2, \Lambda \big )} \right ) \mathbbit{L}_{\mu \nu \rho \sigma} \right ) \! \label{eqn:dressed_graviton_propagator_cc}
\end{equation}
Here, \(\Sigma_\gravprop \big ( p^2, \Lambda \big )\) denotes the renormalized self-energy of the full graviton propagator and \(\Sigma_\gravghostprop \big ( p^2, \Lambda \big )\) denotes the renormalized self-energy of the longitudinal graviton propagator. Both are subjected to the following two renormalization conditions:
\begin{subequations} \label{eqns:ren-conditions-grav-propagator-qgrcc}
\begin{align}
	& \Sigma_\gravprop \big ( - \Lambda, \Lambda \big ) = 0 && \Sigma_\gravghostprop \big ( - \Lambda, \Lambda \big ) = 0
	\intertext{and}
	& \eval{\frac{\partial}{\partial p^2}}_{p^2 = - \Lambda} \Sigma_\gravprop \big ( p^2, \Lambda \big ) = 0 && \eval{\frac{\partial}{\partial p^2}}_{p^2 = - \Lambda} \Sigma_\gravghostprop \big ( p^2, \Lambda \big ) = 0
\end{align}
\end{subequations}
This determines all renormalization constants uniquely. We will discuss these relations in future work. In particular, we remark that with the introduction of the \(Z\)-factor \(Z_\Lambda \! \left ( \varepsilon \right )\) we obtain an additional energy-curvature relation in addition to the Einstein field equations. However, this is only relevant for intermediate calculations and does not contribute to physical scattering amplitudes.

Finally, we discuss the transversality of (Effective) Quantum General Relativity with a cosmological constant in the sense of \cite{Prinz_7}: We find that the graviton propagator does not split into the sum of its transversal and longitudinal projection tensors in \thmref{thm:feynman-rule-graviton-propagator-cc-lt-decomposition} and that the three-valent vertex Feynman rule does not satisfy the contraction identity in \thmref{thm:three-valent-contraction-identities-qgrcc}. Both these identities are satisfied in the case \(\Lambda = 0\). The implication of these results on transversality will be discussed in future work.

\section{Feynman rules}

We start this article by calculating the corresponding Feynman rules as a generalization of \cite{Prinz_4}. In particular, we focus only on the graviton vertex Feynman rules and the graviton propagator: This is due to the fact that the other Feynman rules remain valid with the replacement \(\hat{\eta} \rightsquigarrow \hat{\bgm}\). We refer to \cite{Prinz_2,Prinz_4,Prinz_7} for further information and references.

\enter

\begin{con}
	If the expansion of the gravity-matter Lagrange density is considered with respect to the general background metric \(\bgm_{\mu \nu}\), then we obtain the following additional factors: The vertex Feynman rules are multiplied by \(\sqrt{- \dt{\bgm}}\) and the propagator Feynman rules by \(\textfrac{1}{\sqrt{- \dt{\bgm}}}\). We omit both factors for a cleaner presentation, as diagrammatic Feynman rules do in fact only depend on them via their loop number. More precisely, a Feynman graph with loop number \(L\) obtains an additional factor of \(\sqrt{- \dt{\bgm}}^{L-1}\).
\end{con}

\enter

\begin{thm} \label{thm:grav-fr-cc}
	Given the situation of \cite[Theorem 4.10]{Prinz_4} with the metric decomposition \(g_{\mu \nu} = \bgm_{\mu \nu} + \gcoupling h_{\mu \nu}\), where \(\bgm_{\mu \nu}\) is a suitable background metric. Then the graviton \(2\)-point vertex Feynman rule for (effective) Quantum General Relativity with a cosmological constant reads (where \(\zeta\) denotes the gauge parameter, \(\Lambda\) the cosmological constant and we use momentum conservation on the quadratic term, i.e.\ set \(p_1^\sigma := p^\sigma\) and \(p_2^\sigma := - p^\sigma\)):
	\begin{equation} \label{eqn:quadratic-fr-grav-cc}
	\begin{split}
		\ccgravfr_2^{\mu_1 \nu_1 \vert \mu_2 \nu_2} \left ( p^\sigma, \Lambda; \zeta \right ) & = \frac{\imaginary}{4} \left ( 1 - \frac{1}{\zeta} \right ) \left ( p^{\mu_1} p^{\nu_1} \hat{\bgm}^{\mu_2 \nu_2} + p^{\mu_2} p^{\nu_2} \hat{\bgm}^{\mu_1 \nu_1} \right ) \\
		& \hphantom{=} \mkern-55mu - \frac{\imaginary}{8} \left ( 1 - \frac{1}{\zeta} \right ) \left ( p^{\mu_1} p^{\mu_2} \hat{\bgm}^{\nu_1 \nu_2} + p^{\mu_1} p^{\nu_2} \hat{\bgm}^{\nu_1 \mu_2} + p^{\nu_1} p^{\mu_2} \hat{\bgm}^{\mu_1 \nu_2} + p^{\nu_1} p^{\nu_2} \hat{\bgm}^{\mu_1 \mu_2} \right ) \\
		& \hphantom{=} \mkern-55mu - \frac{\imaginary}{4} \left ( 1 - \frac{1}{2 \zeta} \right ) \left ( p^2 \hat{\bgm}^{\mu_1 \nu_1} \hat{\bgm}^{\mu_2 \nu_2} \right ) \\
		& \hphantom{=} \mkern-55mu + \frac{\imaginary}{8} \left ( p^2 \hat{\bgm}^{\mu_1 \mu_2} \hat{\bgm}^{\nu_1 \nu_2} + p^2 \hat{\bgm}^{\mu_1 \nu_2} \hat{\bgm}^{\nu_1 \mu_2} \right ) \\
		& \hphantom{=} \mkern-55mu + \frac{\imaginary \Lambda}{4} \left ( \hat{\bgm}^{\mu_1 \mu_2} \hat{\bgm}^{\nu_1 \nu_2} + \hat{\bgm}^{\mu_1 \nu_2} \hat{\bgm}^{\nu_1 \mu_2} - \hat{\bgm}^{\mu_1 \nu_1} \hat{\bgm}^{\mu_2 \nu_2} \right )
	\end{split}
	\end{equation}
	Furthermore, the graviton \(n\)-point vertex Feynman rules with \(n > 2\) for (effective) Quantum General Relativity with a cosmological constant read (where \(\gravfr_n\) denotes the corresponding Feynman rules without a cosmological constant, cf.\ \cite[Theorem 4.10]{Prinz_4}, and \(\mathfrak{V}_n\) denotes the corresponding Feynman rules of the Riemannian volume form, cf.\ \cite[Lemma 4.8]{Prinz_4}):
		\begin{equation}
		\begin{split}
			\ccgravfr_n^{\mu_1 \nu_1 \vert \cdots \vert \mu_n \nu_n} \left ( p_1^\sigma, \cdots, p_n^\sigma, \Lambda \right ) & = \eval{\left ( \gravfr_n^{\mu_1 \nu_1 \vert \cdots \vert \mu_n \nu_n} \left ( p_1^\sigma, \cdots, p_n^\sigma \right ) + \frac{\imaginary \Lambda}{\gcoupling^2} \mathfrak{V}_n^{\mu_1 \nu_1 \vert \cdots \vert \mu_n \nu_n} \right )}_{\hat{\eta} \rightsquigarrow \hat{\bgm}}
		\end{split}
		\end{equation}
\end{thm}

\begin{proof}
	This follows directly from the linearity of the Feynman rules together with \cite[Theorem 4.10]{Prinz_4} and \cite[Lemma 4.8]{Prinz_4}. We remark an additional relative minus sign in front of \(\mathfrak{V}_n\) coming from the Fourier transform, as it does not depend on momenta.
\end{proof}

\enter

\begin{thm} \label{thm:grav-prop-cc}
	Given the situation of \thmref{thm:grav-fr-cc}, the graviton propagator Feynman rule for (effective) Quantum General Relativity with a cosmological constant reads:
	\begin{equation} \label{eqn:grav-prop-cc}
	\begin{split}
		\gravprop_{\mu_1 \nu_1 \vert \mu_2 \nu_2} \left ( p^\sigma, \Lambda; \zeta; \epsilon \right ) & = - \frac{2 \imaginary}{p^2 + \Lambda + \imaginary \epsilon} \Bigg [ \left ( \hat{\bgm}_{\mu_1 \mu_2} \hat{\bgm}_{\nu_1 \nu_2} + \hat{\bgm}_{\mu_1 \nu_2} \hat{\bgm}_{\nu_1 \mu_2} - \hat{\bgm}_{\mu_1 \nu_1} \hat{\bgm}_{\mu_2 \nu_2} \right ) \\
		& \! \! \! \! \! \! \! - \left ( \frac{1 - \zeta}{p^2 + \zeta \Lambda} \right ) \left ( \hat{\bgm}_{\mu_1 \mu_2} p_{\nu_1} p_{\nu_2} + \hat{\bgm}_{\mu_1 \nu_2} p_{\nu_1} p_{\mu_2} + \hat{\bgm}_{\nu_1 \mu_2} p_{\mu_1} p_{\nu_2} + \hat{\bgm}_{\nu_1 \nu_2} p_{\mu_1} p_{\mu_2} \right ) \Bigg ]
	\end{split}
	\end{equation}
\end{thm}

\begin{proof}
	To calculate the graviton propagator, we recall the quadratic Feynman rule from \eqnref{eqn:quadratic-fr-grav-cc} and then invert it to obtain the propagator, i.e.\ such that\footnote{Where we treat the tuples of indices \(\mu_i \nu_i\) as one index, i.e.\ exclude the a priori possible term \(\hat{\bgm}^{\mu_1 \nu_1} \hat{\bgm}_{\mu_3 \nu_3}\) on the right hand side.}
	\begin{equation}
		\ccgravfr_2^{\mu_1 \nu_1 \vert \mu_2 \nu_2} \left ( p^\sigma, \Lambda; \zeta \right ) \gravprop_{\mu_2 \nu_2 \vert \mu_3 \nu_3} \left ( p^\sigma, \Lambda; \zeta; 0 \right ) = \frac{1}{2} \left ( \hat{\delta}^{\mu_1}_{\mu_3} \hat{\delta}^{\nu_1}_{\nu_3} + \hat{\delta}^{\mu_1}_{\nu_3} \hat{\delta}^{\nu_1}_{\mu_3} \right )
	\end{equation}
	holds.
\end{proof}

\begin{exmp} \label{exmp:grav-fr-cc}
	Given the situation of \thmref{thm:grav-fr-cc}, the three- and four-valent graviton vertex Feynman rules read as follows:\footnote{We have used momentum conservation, i.e.\ performed a partial integration on the Lagrange density for General Relativity, to obtain a more compact form.}
	\begin{subequations} \label{eqns:grav-fr-cc-triple}
	\begin{align}
		& \ccgravfr_3^{\mu_1 \nu_1 \vert \mu_2 \nu_2 \vert \mu_3 \nu_3} \left ( p_1^\sigma, p_2^\sigma, p_3^\sigma, \Lambda \right ) = \frac{\imaginary}{8} \sum_{\mu_i \leftrightarrow \nu_i} \sum_{s \in S_3} \preccgravfr_3^{\mu_{s(1)} \nu_{s(1)} \vert \mu_{s(2)} \nu_{s(2)} \vert \mu_{s(3)} \nu_{s(3)}} \left ( p_{s(1)}^\sigma, p_{s(2)}^\sigma \right )
		\intertext{with}
		\begin{split}
			& \preccgravfr_3^{\mu_1 \nu_1 \vert \mu_2 \nu_2 \vert \mu_3 \nu_3} \left ( p_1^\sigma, p_2^\sigma, \Lambda \right ) = \frac{\gcoupling}{4} \Bigg \{ \frac{1}{2} p_1^{\mu_3} p_2^{\nu_3} \hat{\bgm}^{\mu_1 \mu_2} \hat{\bgm}^{\nu_1 \nu_2} - p_1^{\mu_3} p_2^{\mu_1} \hat{\bgm}^{\nu_1 \mu_2} \hat{\bgm}^{\nu_2 \nu_3} \\
			& \phantom{\pregravfr_3^{\mu_1 \nu_1 \vert \mu_2 \nu_2 \vert \mu_3 \nu_3} \left ( p_1^\sigma, p_2^\sigma \right ) = \frac{\gcoupling}{4} \Bigg \{} + \left ( p_1 \cdot p_2 \right ) \bigg ( - \frac{1}{2} \hat{\bgm}^{\mu_1 \nu_1 } \hat{\bgm}^{\mu_2 \mu_3 } \hat{\bgm}^{\nu_2 \nu_3 } + \hat{\bgm}^{\mu_1 \nu_2 } \hat{\bgm}^{\mu_2 \nu_3 } \hat{\bgm}^{\mu_3 \nu_1 } \\ & \phantom{\pregravfr_3^{\mu_1 \nu_1 \vert \mu_2 \nu_2 \vert \mu_3 \nu_3} \left ( p_1^\sigma, p_2^\sigma \right ) = \frac{\gcoupling}{4} \Bigg \{ + \left ( p_1 \cdot p_2 \right ) \bigg (} - \frac{1}{4} \hat{\bgm}^{\mu_1 \mu_2 } \hat{\bgm}^{\nu_1 \nu_2 } \hat{\bgm}^{\mu_3 \nu_3 } + \frac{1}{8} \hat{\bgm}^{\mu_1 \nu_1 } \hat{\bgm}^{\mu_2 \nu_2 } \hat{\bgm}^{\mu_3 \nu_3 } \bigg ) \! \Bigg \} \\
			& \phantom{\preccgravfr_3^{\mu_1 \nu_1 \vert \mu_2 \nu_2 \vert \mu_3 \nu_3} \left ( p_1^\sigma, p_2^\sigma \right ) =} + \imaginary \varkappa \Lambda \left \{ \frac{1}{8} \hat{\bgm}^{\mu_1 \nu_1} \hat{\bgm}^{\mu_2 \nu_2} \hat{\bgm}^{\mu_3 \nu_3} - \frac{1}{4} \hat{\bgm}^{\mu_1 \nu_1} \hat{\bgm}^{\mu_2 \nu_3} \hat{\bgm}^{\mu_3 \nu_2} + \hat{\bgm}^{\mu_1 \nu_2} \hat{\bgm}^{\mu_2 \nu_3} \hat{\bgm}^{\mu_3 \nu_1} \right \}
		\end{split}
	\end{align}
	\end{subequations}
	and
	\begin{subequations}
	\begin{align}
		& \ccgravfr_4^{\mu_1 \nu_1 \vert \cdots \vert \mu_4 \nu_4} \left ( p_1^\sigma, \cdots , p_4^\sigma, \Lambda \right ) = \frac{\imaginary}{16} \sum_{\mu_i \leftrightarrow \nu_i} \sum_{s \in S_4} \preccgravfr_4^{\mu_{s(1)} \nu_{s(1)} \vert \cdots \vert \mu_{s(4)} \nu_{s(4)}} \left ( p_{s(1)}^\sigma, p_{s(2)}^\sigma \right )
		\intertext{with}
		\begin{split}
		& \preccgravfr_4^{\mu_1 \nu_1 \vert \cdots \vert \mu_4 \nu_4} \left ( p_1^\sigma, p_2^\sigma, \Lambda \right ) = \frac{\gcoupling}{4} \Bigg \{ - p_1^{\mu_3} p_2^{\nu_3} \hat{\bgm}^{\mu_1 \mu_2} \hat{\bgm}^{\nu_1 \mu_4} \hat{\bgm}^{\nu_2 \nu_4} + 2 p_1^{\mu_3} p_2^{\mu_1} \hat{\bgm}^{\nu_1 \mu_2} \hat{\bgm}^{\nu_2 \mu_4} \hat{\bgm}^{\nu_3 \nu_4} \\
		& \phantom{\pregravfr_4^{\mu_1 \nu_1 \vert \cdots \vert \mu_4 \nu_4} \left ( p_1^\sigma, p_2^\sigma \right ) = \frac{\gcoupling}{4} \Bigg \{} - \frac{1}{2} p_1^{\mu_3} p_2^{\mu_1} \hat{\bgm}^{\nu_1 \mu_2} \hat{\bgm}^{\nu_2 \nu_3} \hat{\bgm}^{\mu_4 \nu_4} + p_1^{\mu_3} p_2^{\nu_3} \hat{\bgm}^{\mu_1 \mu_2} \hat{\bgm}^{\nu_1 \nu_2} \hat{\bgm}^{\mu_4 \nu_4} \\
		& \phantom{\pregravfr_4^{\mu_1 \nu_1 \vert \cdots \vert \mu_4 \nu_4} \left ( p_1^\sigma, p_2^\sigma \right ) = \frac{\gcoupling}{4} \Bigg \{} - \frac{1}{2} p_1^{\mu_3} p_2^{\mu_4} \hat{\bgm}^{\mu_1 \mu_2} \hat{\bgm}^{ \nu_1 \nu_2} \hat{\bgm}^{\nu_3 \nu_4} + p_1^{\mu_3} p_2^{\mu_4} \hat{\bgm}^{\mu_1 \mu_2} \hat{\bgm}^{\nu_1 \nu_3} \hat{\bgm}^{\nu_2 \nu_4} \\
		& \phantom{\pregravfr_4^{\mu_1 \nu_1 \vert \cdots \vert \mu_4 \nu_4} \left ( p_1^\sigma, p_2^\sigma \right ) = \frac{\gcoupling}{4} \Bigg \{} - \frac{1}{2} p_1^{\mu_2} p_2^{\mu_3} \hat{\bgm}^{\mu_1 \nu_1} \hat{\bgm}^{\nu_2 \mu_4} \hat{\bgm}^{\nu_3 \nu_4} + \frac{1}{4} p_1^{\mu_2} p_2^{\mu_1} \hat{\bgm}^{\nu_1 \nu_2} \hat{\bgm}^{\mu_3 \mu_4} \hat{\bgm}^{\nu_3 \nu_4} \\
		& \phantom{\pregravfr_4^{\mu_1 \nu_1 \vert \cdots \vert \mu_4 \nu_4} \left ( p_1^\sigma, p_2^\sigma \right ) =} \! \! \! \! + \left ( p_1 \cdot p_2 \right ) \bigg ( - \frac{1}{16} \hat{\bgm}^{\mu_1 \mu_2} \hat{\bgm}^{\nu_1 \nu_2} \hat{\bgm}^{\mu_3 \nu_3} \hat{\bgm}^{\mu_4 \nu_4} + \frac{1}{8} \hat{\bgm}^{\mu_1 \mu_2} \hat{\bgm}^{\nu_1 \nu_2} \hat{\bgm}^{\mu_3 \mu_4} \hat{\bgm}^{\nu_3 \nu_4} \\
		& \phantom{\pregravfr_4^{\mu_1 \nu_1 \vert \cdots \vert \mu_4 \nu_4} \left ( p_1^\sigma, p_2^\sigma \right ) = + \left ( p_1 \cdot p_2 \right ) \bigg (} \! \! \! \! + \frac{1}{2} \hat{\bgm}^{\mu_1 \mu_2} \hat{\bgm}^{\nu_1 \mu_3} \hat{\bgm}^{\nu_2 \nu_3} \hat{\bgm}^{\mu_4 \nu_4} - \hat{\bgm}^{\mu_1 \mu_2} \hat{\bgm}^{\nu_1 \mu_3} \hat{\bgm}^{\nu_2 \mu_4} \hat{\bgm}^{\nu_3 \nu_4} \\
		& \phantom{\pregravfr_4^{\mu_1 \nu_1 \vert \cdots \vert \mu_4 \nu_4} \left ( p_1^\sigma, p_2^\sigma \right ) = + \left ( p_1 \cdot p_2 \right ) \bigg (} \! \! \! \! + \frac{1}{2} \hat{\bgm}^{\mu_1 \mu_3} \hat{\bgm}^{\nu_1 \nu_3} \hat{\bgm}^{\mu_2 \mu_4} \hat{\bgm}^{\nu_2 \nu_4} - \frac{1}{2} \hat{\bgm}^{\mu_1 \mu_3} \hat{\bgm}^{\nu_1 \mu_4} \hat{\bgm}^{\mu_2 \nu_3} \hat{\bgm}^{\nu_2 \nu_4} \\
		& \phantom{\pregravfr_4^{\mu_1 \nu_1 \vert \cdots \vert \mu_4 \nu_4} \left ( p_1^\sigma, p_2^\sigma \right ) = + \left ( p_1 \cdot p_2 \right ) \bigg (} \! \! \! \! - \frac{1}{4} \hat{\bgm}^{\mu_1 \nu_1} \hat{\bgm}^{\mu_2 \mu_3} \hat{\bgm}^{\nu_2 \nu_3} \hat{\bgm}^{\mu_4 \nu_4} + \frac{1}{2} \hat{\bgm}^{\mu_1 \nu_1} \hat{\bgm}^{\mu_2 \mu_3} \hat{\bgm}^{\nu_2 \mu_4} \hat{\bgm}^{\nu_3 \nu_4} \\
		& \phantom{\pregravfr_4^{\mu_1 \nu_1 \vert \cdots \vert \mu_4 \nu_4} \left ( p_1^\sigma, p_2^\sigma \right ) = + \left ( p_1 \cdot p_2 \right ) \bigg (} \! \! \! \! + \frac{1}{32} \hat{\bgm}^{\mu_1 \nu_1} \hat{\bgm}^{\mu_2 \nu_2} \hat{\bgm}^{\mu_3 \nu_3} \hat{\bgm}^{\mu_4 \nu_4} - \frac{1}{8} \hat{\bgm}^{\mu_1 \nu_1} \hat{\bgm}^{\mu_2 \nu_2} \hat{\bgm}^{\mu_3 \mu_4} \hat{\bgm}^{\nu_3 \nu_4} \bigg ) \Bigg \} \\
			& \phantom{\preccgravfr_4^{\mu_1 \nu_1 \vert \cdots \vert \mu_4 \nu_4} \left ( p_1^\sigma, p_2^\sigma \right ) =} + \imaginary \varkappa \Lambda \, \bigg \{ \frac{57}{16} \hat{\bgm}^{\mu_1 \nu_1} \hat{\bgm}^{\mu_2 \nu_2} \hat{\bgm}^{\mu_3 \nu_3} \hat{\bgm}^{\mu_4 \nu_4} + \frac{15}{2} \hat{\bgm}^{\mu_1 \nu_1} \hat{\bgm}^{\mu_2 \nu_3} \hat{\bgm}^{\mu_3 \nu_4} \hat{\bgm}^{\mu_4 \nu_2} \\
			& \phantom{\preccgravfr_4^{\mu_1 \nu_1 \vert \cdots \vert \mu_4 \nu_4} \left ( p_1^\sigma, p_2^\sigma \right ) = + \imaginary \varkappa \Lambda \bigg \{} - \frac{13}{8} \hat{\bgm}^{\mu_1 \nu_1} \hat{\bgm}^{\mu_2 \nu_2} \hat{\bgm}^{\mu_3 \nu_4} \hat{\bgm}^{\mu_4 \nu_3} + \frac{7}{8} \hat{\bgm}^{\mu_1 \nu_2} \hat{\bgm}^{\mu_2 \nu_1} \hat{\bgm}^{\mu_3 \nu_4} \hat{\bgm}^{\mu_4 \nu_3} \\
			& \phantom{\preccgravfr_4^{\mu_1 \nu_1 \vert \cdots \vert \mu_4 \nu_4} \left ( p_1^\sigma, p_2^\sigma \right ) = + \imaginary \varkappa \Lambda \bigg \{} - 12 \hat{\bgm}^{\mu_1 \nu_2} \hat{\bgm}^{\mu_2 \nu_3} \hat{\bgm}^{\mu_3 \nu_4} \hat{\bgm}^{\mu_4 \nu_1} \bigg \}
		\end{split}
	\end{align}
	\end{subequations}
	We remark that the three- and four-valent graviton vertex Feynman rules agree with the cited literature modulo prefactors and minus signs. Additionally, we remark that the three- and four-valent graviton-ghost Feynman rules are given in \cite[Example 4.15]{Prinz_4} and remain valid with the replacement \(\hat{\eta} \rightsquigarrow \hat{\bgm}\).
\end{exmp}

\section{BRST setup, Faddeev--Popov ghosts and symmetric ghosts}

In this section, we discuss the generalization of the results from \cite{Prinz_5,Prinz_6}: This includes the diffeomorphism BRST and anti-BRST operators, the Faddeev--Popov ghost Lagrange density as well as the symmetric ghost Lagrange density and their relation. More precisely, the diffeomorphism BRST operator \(P\) is given as the following odd vector field on the spacetime-matter bundle with graviton-ghost degree 1:
	\begin{equation}
		P := \frac{1}{\zeta} \left ( \nabla^{TM}_\mu C_\nu + \nabla^{TM}_\nu C_\mu \right ) \frac{\partial}{\partial h_{\mu \nu}} + \varkappa C^\rho \left ( \partial_\rho C^\sigma \right ) \frac{\partial}{\partial C^\sigma} + \frac{1}{\zeta} B^\sigma \frac{\partial}{\partial \overline{C}^\sigma}
	\end{equation}
	Equivalently, its action on fundamental particle fields is given as follows:
	{\allowdisplaybreaks
	\begin{subequations}
	\begin{align}
		P h_{\mu \nu} & := \frac{1}{\zeta} \left ( \nabla^{TM}_\mu C_\nu + \nabla^{TM}_\nu C_\mu \right ) \\
		P C^\rho & := \varkappa C^\sigma \left ( \partial_\sigma C^\rho \right ) \\
		P \overline{C}^\rho & := \frac{1}{\zeta} B^\rho \\
		P B^\rho & := 0 \\
		P \bgm_{\mu \nu} & := 0 \label{eqn:diffeomorphism-brst-operator-bgm}
	\end{align}
	\end{subequations}
	}%
	Additionally, we define the diffeomorphism anti-BRST operator \(\overline{P}\) as the following odd vector field on the spacetime-matter bundle with graviton-ghost degree -1:
	\begin{subequations}
	\begin{align}
		\overline{P} & := \eval{P}_{C \rightsquigarrow \overline{C}}
		\intertext{together with the following additional changes}
		\overline{P} C^\rho & := - \frac{1}{\zeta} B^\rho + \varkappa \left ( \overline{C}^\sigma \left ( \partial_\sigma C^\rho \right ) - \big ( \partial_\sigma \overline{C}^\rho \big ) C^\sigma \right ) \\
		\overline{P} \overline{C}^\rho & := \varkappa \overline{C}^\sigma \big ( \partial_\sigma \overline{C}_\rho \big ) \\
		\overline{P} B^\rho & := \varkappa \left ( \overline{C}^\sigma \left ( \partial_\sigma B_\rho \right ) - \big ( \partial_\sigma \overline{C}^\rho \big ) B^\sigma \right )
	\end{align}
	\end{subequations}
	We remark the characteristic identities \(\commutatorbig{P}{P} = \commutatorbig{P}{\overline{P}} = \commutatorbig{\overline{P}}{\overline{P}} = 0\),\footnote{We emphasize that \(\left [ \cdot , \cdot \right ]\) denotes the supercommutator, which is equivalent to the anticommutator in the previous equation.} cf.\ \cite[Proposition 3.2 and Corollary 3.8]{Prinz_5} and the references therein.

\enter

\begin{rem}
	The definition of the diffeomorphism BRST operator given above is identical to the case of a flat background metric, cf.\ \cite[Definition 3.1 and Definition 3.7]{Prinz_5} and \cite[Definition 2.2]{Prinz_6} and the references therein: The crucial point is that the background metric \(\eta_{\mu \nu}\) or \(\bgm_{\mu \nu}\) is set to be invariant with respect to \(P\), as in \eqnref{eqn:diffeomorphism-brst-operator-bgm}. This is achieved by defining its transformation into the graviton field \(h_{\mu \nu}\).
\end{rem}

\enter

\begin{col} \label{col:linearized_de_donder_gauge_fixing_fermion-cc}
	The gauge fixing Lagrange density and its accompanying Faddeev--Popov ghost Lagrange density for (effective) Quantum General Relativity with a cosmological constant read
	\begin{equation}
	\begin{split}
		\mathcal{L}_\textup{GR-\(\Lambda\)-GF} + \mathcal{L}_\textup{GR-\(\Lambda\)-Ghost} & = - \frac{1}{4 \varkappa^2 \zeta}  \bgm^{\mu \nu} \deDonder^{(1)}_\mu \deDonder^{(1)}_\nu \dif V_\bgm \\ & \phantom{=} - \frac{1}{2 \zeta} \bgm^{\rho \sigma} \overline{C}^\mu \left ( \partial_\rho \partial_\sigma C_\mu \right ) \dif V_\bgm \\ & \phantom{=} - \frac{1}{2} \bgm^{\rho \sigma} \overline{C}^\mu \left ( \partial_\mu \big ( \tensor{\Gamma}{^\nu _\rho _\sigma} C_\nu \big ) - 2 \partial_\rho \big ( \tensor{\Gamma}{^\nu _\mu _\sigma} C_\nu \big ) \right ) \dif V_\bgm
	\end{split}
	\end{equation}
	with the linearized de Donder gauge fixing functional \(\deDonder^{(1)}_\mu := \bgm^{\rho \sigma} \Gamma_{\rho \sigma \mu}\). They can be obtained from the gauge fixing fermion
	\begin{equation}
		\stigma^{(1)} := \frac{1}{2} \overline{C}^\rho \left ( \frac{1}{\varkappa} \deDonder^{(1)}_\rho + \frac{1}{2} B_\rho \right ) \dif V_\bgm \label{eqn:linearized_de_donder_gauge_fixing_fermion-cc}
	\end{equation}
	via \(P \stigma^{(1)}\).
\end{col}

\begin{proof}
	This can be shown analogously to the proof of \cite[Proposition 3.4]{Prinz_5}.
\end{proof}

\enter

\begin{col} \label{col:symmetric-ghost-qgrcc}
	The gauge fixing Lagrange density and its accompanying symmetric ghost Lagrange density for (effective) Quantum General Relativity with a cosmological constant read
	\begin{equation}
	\begin{split}
		\mathcal{L}_\textup{GR-\(\Lambda\)-GF-Ghost} & = - \frac{1}{2 \zeta} \left ( \frac{1}{2 \varkappa^2} \bgm^{\mu \nu} \deDonder^{(1)}_\mu \deDonder^{(1)}_\nu + \bgm^{\mu \nu} \big ( \partial_\mu \overline{C}^\rho \big ) \big ( \partial_\nu C_\rho \big ) \right ) \dif V_\bgm \\
		& \phantom{:=} - \frac{1}{2} \bgm^{\mu \nu} \overline{C}^\rho \left ( \frac{1}{2} \partial_\rho \big ( \tensor{\Gamma}{^\sigma _\mu _\nu} C_\sigma \big ) - \partial_\mu \big ( \tensor{\Gamma}{^\sigma _\rho _\nu} C_\sigma \big ) \right ) \dif V_\bgm \\
		& \phantom{:=} + \frac{1}{2} \bgm^{\mu \nu} \left ( \frac{1}{2} \partial_\rho \big ( \tensor{\Gamma}{^\sigma _\mu _\nu} \overline{C}_\sigma \big ) - \partial_\mu \big ( \tensor{\Gamma}{^\sigma _\rho _\nu} \overline{C}_\sigma \big ) \right ) C^\rho \dif V_\bgm \\
		& \phantom{:=} - \frac{\varkappa^2 \zeta}{8} \bgm_{\mu \nu} \left ( \overline{C}^\rho \big ( \partial_\rho \overline{C}^\mu \big ) \right ) \left ( C^\sigma \big ( \partial_\sigma C^\nu \big ) \right ) \dif V_\bgm
	\end{split}
	\end{equation}
	with the linearized de Donder gauge fixing functional \(\deDonder^{(1)}_\mu := \bgm^{\rho \sigma} \Gamma_{\rho \sigma \mu}\). It can be obtained from the gauge fixing boson
	\begin{equation}
		F := \frac{\zeta}{4} \left ( \frac{1}{\varkappa} \bgm^{\mu \nu} h_{\mu \nu} + \overline{C}^\rho C_\rho \right ) \dif V_\bgm
	\end{equation}
	via \(\mathcal{P} F\), where \(\mathcal{P} := P \circ \overline{P}\) is the diffeomorphism super-BRST operator, cf.\ \cite[Definition 2.5]{Prinz_6}.
\end{col}

\begin{proof}
	This can be shown analogously to the proof of \cite[Proposition 3.1]{Prinz_6}.
\end{proof}

\enter

\begin{col} \label{col:symmetric-ghost-homotopy-qgrcc}
	We obtain the following homotopy in \(\lambda \in [0, 1]\) between the Faddeev--Popov construction \(\lambda = 0\), the symmetric setup \(\lambda = \textfrac{1}{2}\) and the opposed Faddeev--Popov construction \(\lambda = 1\):
	\begin{equation}
	\begin{split}
	\mathcal{L}_\textup{GR-\(\Lambda\)-GF-Ghost} \left ( \lambda \right ) & = - \frac{1}{2 \zeta} \left ( \frac{1}{2 \varkappa^2} \bgm^{\mu \nu} \deDonder^{(1)}_\mu \deDonder^{(1)}_\nu + \bgm^{\mu \nu} \big ( \partial_\mu \overline{C}^\rho \big ) \big ( \partial_\nu C_\rho \big ) \right ) \dif V_\bgm \\
	& \phantom{=} - \left ( 1 - \lambda \right ) \bgm^{\mu \nu} \overline{C}^\rho \left ( \frac{1}{2} \partial_\rho \big ( \tensor{\Gamma}{^\sigma _\mu _\nu} C_\sigma \big ) - \partial_\mu \big ( \tensor{\Gamma}{^\sigma _\rho _\nu} C_\sigma \big ) \right ) \dif V_\bgm \\
	& \phantom{=} + \lambda \bgm^{\mu \nu} \left ( \frac{1}{2} \partial_\rho \big ( \tensor{\Gamma}{^\sigma _\mu _\nu} \overline{C}_\sigma \big ) - \partial_\mu \big ( \tensor{\Gamma}{^\sigma _\rho _\nu} \overline{C}_\sigma \big ) \right ) C^\rho \dif V_\bgm \\
	& \phantom{=} - \frac{\varkappa^2 \zeta}{8} \lambda \left ( 1 - \lambda \right ) \bgm_{\mu \nu} \left ( \overline{C}^\rho \big ( \partial_\rho \overline{C}^\mu \big ) \right ) \left ( C^\sigma \big ( \partial_\sigma C^\nu \big ) \right ) \dif V_\bgm
	\end{split}
	\end{equation}
\end{col}

\begin{proof}
	This can be shown analogously to the proof of \cite[Theorem 3.2]{Prinz_6}.
\end{proof}

\section{Transversality}

Finally, we study the transversality of (effective) Quantum General Relativity with a cosmological constant as a generalization of \cite{Prinz_7}. Therein, in Definition 3.14, we introduced the transversal structure of (effective) Quantum General Relativity with a de Donder gauge fixing (QGR) as the set \(\mathcal{T}_\text{QGR} := \setbig{\mathbbit{L}, \mathbbit{I}, \mathbbit{T}}\). These three operators are the longitudinal, identical and transversal projection tensors, given in our generalized setting as follows:\footnote{From now on, we omit the hat for Fourier transformed quantities to improve readability.}
\begin{subequations} \label{eqn:defn_projection_tensors_qgrcc}
\begin{align}
	\mathbbit{L}^{\rho \sigma}_{\mu \nu} & := \frac{1}{2 p^2} \left ( \delta^\rho_\mu p^\sigma p_\nu + \delta^\sigma_\mu p^\rho p_\nu + \delta^\rho_\nu p^\sigma p_\mu + \delta^\sigma_\nu p^\rho p_\mu - 2 \bgm^{\rho \sigma} p_\mu p_\nu \right ) \, , \\
	\mathbbit{I} \mspace{2mu} ^{\rho \sigma}_{\mu \nu} & := \frac{1}{2} \left ( \delta^\rho_\mu \delta^\sigma_\nu + \delta^\sigma_\mu \delta^\rho_\nu \right )
	\intertext{and}
	\mathbbit{T} \mspace{2mu} ^{\rho \sigma}_{\mu \nu} & := \mathbbit{I} \mspace{2mu} ^{\rho \sigma}_{\mu \nu} - \mathbbit{L}^{\rho \sigma}_{\mu \nu} \, ,
\end{align}
\end{subequations}
where we have set \(p^2 := \bgm_{\mu \nu} p^\mu p^\nu\). Lorentz indices on \(\mathbbit{L}\), \(\mathbbit{I}\) and \(\mathbbit{T}\) are raised and lowered with the metric \(\mathbbit{G}\), defined via\footnote{This metric is known in the literature as the de Witt metric \cite{DeWitt_I}, cf.\ \cite{Giulini_Kiefer}. The reason for the asymmetric definition concerning the factor \(\textfrac{1}{4}\) is motivated by \cite[Equations (57)]{Prinz_7}.}
\begin{subequations} \label{eqn:defn_metric_ts-qgrcc}
\begin{align}
	\mathbbit{G}_{\mu \nu \rho \sigma} & := \frac{1}{p^2} \left ( \bgm_{\mu \rho} \bgm_{\nu \sigma} + \bgm_{\mu \sigma} \bgm_{\nu \rho} - \bgm_{\mu \nu} \bgm_{\rho \sigma} \right )
	\intertext{and its inverse}
	\mathbbit{G}^{\mu \nu \rho \sigma} & := \frac{p^2}{4} \left ( \bgm^{\mu \rho} \bgm^{\nu \sigma} + \bgm^{\mu \sigma} \bgm^{\nu \rho} - \bgm^{\mu \nu} \bgm^{\rho \sigma} \right ) \, .
\end{align}
\end{subequations}
All these tensors share interesting relations, which are listed in \cite[Subsection 3.2]{Prinz_7}.

\enter

\begin{thm} \label{thm:feynman-rule-graviton-propagator-cc-lt-decomposition}
	The Feynman rule of the graviton propagator with a cosmological constant can be written as follows:
	\begin{equation}
		\Phi \left ( \cgreen{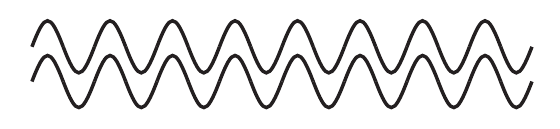} \right ) = - \frac{2 \imaginary p^2}{p^2 + \Lambda + \imaginary \epsilon} \left ( \mathbbit{G}_{\mu \nu \rho \sigma} - \left ( \frac{1 - \zeta}{p^2 + \zeta \Lambda} \right ) \mathbbit{L}_{\mu \nu \rho \sigma} \right ) \label{eqn:decomposition_graviton_propagator_cc}
	\end{equation}
	In particular, we obtain the following two identities:
	\begin{subequations}
	\begin{align}
		\eval{\Phi \left ( \cgreen{p-graviton} \right )}_{\zeta = 0} = - \frac{2 \imaginary p^2}{p^2 + \Lambda + \imaginary \epsilon} \mathbbit{T}_{\mu \nu \rho \sigma}
		\intertext{and}
		\eval{\Phi \left ( \cgreen{p-graviton} \right )}_{\zeta = 1} = - \frac{2 \imaginary p^2}{p^2 + \Lambda + \imaginary \epsilon} \mathbbit{G}_{\mu \nu \rho \sigma}
	\end{align}
	\end{subequations}
\end{thm}

\begin{proof}
	This follows directly from \thmref{thm:grav-prop-cc} and \cite[Definition 3.14]{Prinz_7}, reproduced in \eqnsaref{eqn:defn_projection_tensors_qgrcc}{eqn:defn_metric_ts-qgrcc}.
\end{proof}

\enter

\begin{rem}
	In particular, we see from \thmref{thm:feynman-rule-graviton-propagator-cc-lt-decomposition} that the de Donder gauge fixing is not the optimal gauge fixing condition for (effective) Quantum General Relativity with a cosmological constant, cf.\ \cite[Definition 3.2]{Prinz_7}. This comes from the fact that the cosmological constant behaves in the propagator like an Euclidean mass term and thus alters the longitudinal and transversal modes. Interestingly, the transversal behavior for the two cases \(\zeta = 0\) and \(\zeta = 1\) agrees with the graviton propagator without a cosmological constant, cf.\ \cite[Theorem 3.21]{Prinz_7}.
\end{rem}

\enter

\begin{thm} \label{thm:three-valent-contraction-identities-qgrcc}
	The Feynman rule of the three-valent graviton vertex with a cosmological constant does not satisfy the contraction identity, i.e.\ we obtain:
	\begin{equation}
		\Phi \left ( \bbsL \tcgreen{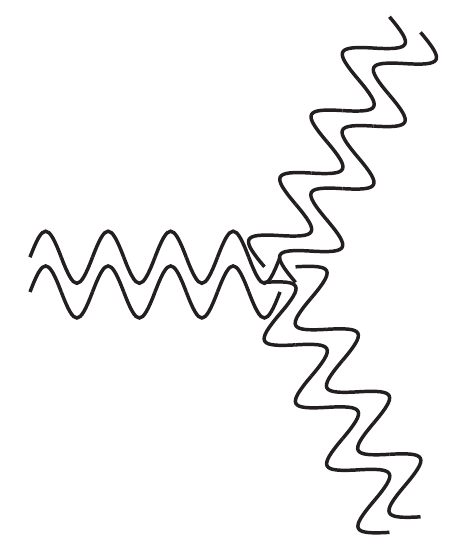} _{\bbsL}^{\bbsL} \right ) \not \simeq_\textup{MC} 0 \label{eqn:contraction_v-triple-graviton-cc} \, ,
	\end{equation}
	where \(\simeq_\textup{MC}\) indicates equality modulo momentum conservation.
\end{thm}

\begin{proof}
	This is verified by the following calculation: First, we use the following decomposition of \(\mathbbit{L}^{\rho \sigma}_{\mu \nu} = \mathscr{L}^{\rho \sigma}_\tau \mathscr{G}_{\mu \nu}^\tau\) into the two tensors \(\mathscr{G}_{\mu \nu}^\kappa := \frac{1}{p^2} \big ( p_\mu \delta_\nu^\kappa + p_\nu \delta_\mu^\kappa \big )\) and \(\mathscr{L}^{\rho \sigma}_\lambda := \frac{1}{2} \big ( p^\rho \delta^\sigma_\lambda + p^\sigma \delta^\rho_\lambda - p_\lambda \bgm^{\rho \sigma} \big )\) that has been established in \cite[Lemma 3.16]{Prinz_7}. Then, we recall from \cite[Theorem 3.23]{Prinz_7}
	\begin{align}
	\begin{split}
		\eval{\Phi \left ( \scriptstyle{\mathscr{G}} \tcgreen{v-gravitontriple} _{\scriptstyle{\mathscr{G}}}^{\scriptstyle{\mathscr{G}}} \right )}_{\ordersb{\Lambda^0}} & \simeq_\textup{MC} 0 \, .
	\end{split}
	\intertext{Thus, using \eqnsref{eqns:grav-fr-cc-triple}, the claimed result follows from}
	\begin{split}
		\eval{\Phi \left ( \scriptstyle{\mathscr{G}} \tcgreen{v-gravitontriple} _{\scriptstyle{\mathscr{G}}}^{\scriptstyle{\mathscr{G}}} \right )}_{\ordersb{\Lambda^1}} & \simeq_\textup{MC} - \frac{\imaginary \varkappa \Lambda}{p_1^2 p_2^2 \left ( p_1 + p_2 \right )^2} \bigg ( \bgm^{\tau_1 \tau_2} \left ( - p_1^2 p_2^{\tau_3} - p_2^2 p_1^{\tau_3} \right ) \\
		& \phantom{\simeq_\textup{MC} - \frac{\imaginary \varkappa \Lambda}{p_1^2 p_2^2 \left ( p_1 + p_2 \right )^2} \bigg (} + \bgm^{\tau_1 \tau_3} \left ( p_1^2 \left ( p_1^{\tau_2} + p_2^{\tau_2} \right ) - \left ( p_1 + p_2 \right )^2 p_1^{\tau_2} \right ) \\
		& \phantom{\simeq_\textup{MC} - \frac{\imaginary \varkappa \Lambda}{p_1^2 p_2^2 \left ( p_1 + p_2 \right )^2} \bigg (} + \bgm^{\tau_2 \tau_3} \left ( p_2^2 \left ( p_1^{\tau_1} + p_2^{\tau_1} \right ) - \left ( p_1 + p_2 \right )^2 p_2^{\tau_1} \right ) \! \bigg ) \\
		& \not \simeq_\textup{MC} 0 \, ,
	\end{split}
	\end{align}
	which does only vanish on-shell.
\end{proof}

\enter

\begin{rem} \label{rem:contraction-single-graviton}
	\eqnref{eqn:contraction_v-triple-graviton-cc} indicates that the question of transversality is more involved, if a cosmological constant is added to (effective) Quantum General Relativity. This will be studied in future work.
\end{rem}

\section{Conclusion} \label{sec:conclusion}

We have studied (effective) Quantum General Relativity with a cosmological constant from the viewpoint of renormalization. To this end, we started with its Lagrange density in \eqnsref{eqns:lagrange-density-qgrcc} and discussed the corresponding quantum gauge symmetries in \eqnsref{eqns:z-factor_identities_qgrcc} together with appropriate renormalization conditions for the graviton propagator in \eqnsref{eqns:ren-conditions-grav-propagator-qgrcc}, cf.\ \cite{Prinz_3} for the situation of a vanishing cosmological constant. Then, we calculated the corresponding graviton vertex Feynman rules for arbitrary valence in \thmref{thm:grav-fr-cc} and the graviton propagator Feynman rule for a general gauge parameter in \thmref{thm:grav-prop-cc}. Additionally, we displayed the three-valent and four-valent vertex Feynman rules explicitly in \exref{exmp:grav-fr-cc}, cf.\ \cite{Prinz_4} for the situation of a vanishing cosmological constant. Next, we constructed the Faddeev--Popov ghost, symmetric ghost and homotopy ghost Lagrange densities in \colssaref{col:linearized_de_donder_gauge_fixing_fermion-cc}{col:symmetric-ghost-qgrcc}{col:symmetric-ghost-homotopy-qgrcc} from the viewpoint of BRST cohomology. We refer to \cite{Prinz_5,Prinz_6} for the situation of a vanishing cosmological constant. In addition, we studied its transversality, in particular the decomposition of the graviton propagator Feynman rule in \thmref{thm:feynman-rule-graviton-propagator-cc-lt-decomposition} and the contraction identity of the three-valent graviton vertex Feynman rule in \thmref{thm:three-valent-contraction-identities-qgrcc}, cf.\ \cite{Prinz_7} for the situation of a vanishing cosmological constant.

Finally, we want to mention a different view on the cosmological constant: The possibility to use it as a diffeomorphism-invariant regulator for pure gravity. In this case, a convenient choice would be \(\Lambda \in \mathbb{C}\) and \(\bgm_{\mu \nu} := \eta_{\mu \nu}\), so that the final result (after the limit \(\Lambda \mapsto 0\)) can be considered on the Minkowski background spacetime. This claim is due to the fact that (effective) Quantum General Relativity with only \(\Lambda\)-vertices is renormalizable by power counting. This will be studied in future work.

\section*{Acknowledgments}
\addcontentsline{toc}{section}{Acknowledgments}

This article was initiated by a question asked by my supervisor Dirk Kreimer during my disputation. Having said this, it is my pleasure to express my gratitude for his insight and support over the past years! This research is supported by the \emph{Riemann Center for Geometry and Physics} of the Leibniz University Hannover via a \emph{Riemann fellowship} in the group of Domenico Giulini. It is my pleasure to thank Domenico Giulini, Philip Schwartz and Lorenzo Casarin for a welcoming atmosphere and an inspiring intellectual exchange.

\bibliography{References}{}

\begin{thebibliography}{10}
  \providebibliographyfont{name}{}%
  \providebibliographyfont{lastname}{}%
  \providebibliographyfont{title}{\emph}%
  \providebibliographyfont{jtitle}{\btxtitlefont}%
  \providebibliographyfont{etal}{\emph}%
  \providebibliographyfont{journal}{}%
  \providebibliographyfont{volume}{}%
  \providebibliographyfont{ISBN}{\MakeUppercase}%
  \providebibliographyfont{ISSN}{\MakeUppercase}%
  \providebibliographyfont{url}{\url}%
  \providebibliographyfont{numeral}{}%
  \expandafter\btxselectlanguage\expandafter {\btxfallbacklanguage}

\btxselectlanguage {english}
\bibitem {tHooft_Veltman_OLD}
\btxnamefont {\btxlastnamefont {{G. 't Hooft}}} \btxandlong {}\ \btxnamefont
  {\btxlastnamefont {{M. J. G. Veltman}}}\btxauthorcolon\ \btxjtitlefont
  {\btxifchangecase {{O}ne loop divergencies in the theory of
  gravitation}{{O}ne loop divergencies in the theory of gravitation}}.
\newblock \btxjournalfont {Ann. Inst. H. Poincare Phys. Theor. A}, 20:69--94,
  1974.

\bibitem {Goroff_Sagnotti_TLD}
\btxnamefont {\btxlastnamefont {{M. H. Goroff}}} \btxandlong {}\ \btxnamefont
  {\btxlastnamefont {{A. Sagnotti}}}\btxauthorcolon\ \btxjtitlefont
  {\btxifchangecase {{Q}uantum {G}ravity at {T}wo {L}oops}{{Q}uantum {G}ravity
  at {T}wo {L}oops}}.
\newblock \btxjournalfont {Phys. Lett. B}, 160:81--86, 1985.

\bibitem {Abbott_BFM}
\btxnamefont {\btxlastnamefont {{L. F. Abbott}}}\btxauthorcolon\ \btxjtitlefont
  {\btxifchangecase {{I}ntroduction to the {B}ackground {F}ield
  {M}ethod}{{I}ntroduction to the {B}ackground {F}ield {M}ethod}}.
\newblock \btxjournalfont {Acta Phys. Polon. B}, 13:33, 1982.

\bibitem {Krasnov}
\btxnamefont {\btxlastnamefont {{K. Krasnov}}}\btxauthorcolon\ \btxjtitlefont
  {\btxifchangecase {{E}ffective metric {L}agrangians from an underlying theory
  with two propagating degrees of freedom}{{E}ffective metric {L}agrangians
  from an underlying theory with two propagating degrees of freedom}}.
\newblock \btxjournalfont {Phys. Rev. D}, 81:084026, 2010.
\newblock arXiv:0911.4903v1 [hep-th].

\bibitem {Kreimer_Hopf_Algebra}
\btxnamefont {\btxlastnamefont {{D. Kreimer}}}\btxauthorcolon\ \btxjtitlefont
  {\btxifchangecase {{O}n the {H}opf algebra structure of perturbative quantum
  field theories}{{O}n the {H}opf algebra structure of perturbative quantum
  field theories}}.
\newblock \btxjournalfont {Adv. Theor. Math. Phys.}, 2:303--334, 1998.
\newblock arXiv:q-alg/9707029v4.

\bibitem {Connes_Kreimer_NG}
\btxnamefont {\btxlastnamefont {{A. Connes}}} \btxandlong {}\ \btxnamefont
  {\btxlastnamefont {{D. Kreimer}}}\btxauthorcolon\ \btxjtitlefont
  {\btxifchangecase {{H}opf {A}lgebras, {R}enormalization and {N}oncommutative
  {G}eometry}{{H}opf {A}lgebras, {R}enormalization and {N}oncommutative
  {G}eometry}}.
\newblock \btxjournalfont {Commun. Math. Phys.}, 199:203--242, 1998.
\newblock arXiv:hep-th/9808042v1.

\bibitem {Connes_Kreimer_0}
\btxnamefont {\btxlastnamefont {{A. Connes}}} \btxandlong {}\ \btxnamefont
  {\btxlastnamefont {{D. Kreimer}}}\btxauthorcolon\ \btxjtitlefont
  {\btxifchangecase {{R}enormalization in quantum field theory and the
  {R}iemann-{H}ilbert problem}{{R}enormalization in quantum field theory and
  the {R}iemann-{H}ilbert problem}}.
\newblock \btxjournalfont {JHEP}, 09:024, 1999.
\newblock arXiv:hep-th/9909126v3.

\bibitem {Connes_Kreimer_1}
\btxnamefont {\btxlastnamefont {{A. Connes}}} \btxandlong {}\ \btxnamefont
  {\btxlastnamefont {{D. Kreimer}}}\btxauthorcolon\ \btxjtitlefont
  {\btxifchangecase {{R}enormalization in quantum field theory and the
  {R}iemann-{H}ilbert problem {I}: the {H}opf algebra structure of graphs and
  the main theorem}{{R}enormalization in quantum field theory and the
  {R}iemann-{H}ilbert problem {I}: the {H}opf algebra structure of graphs and
  the main theorem}}.
\newblock \btxjournalfont {Commun. Math. Phys.}, 210:249--273, 1999.
\newblock arXiv:hep-th/9912092v1.

\bibitem {Connes_Kreimer_2}
\btxnamefont {\btxlastnamefont {{A. Connes}}} \btxandlong {}\ \btxnamefont
  {\btxlastnamefont {{D. Kreimer}}}\btxauthorcolon\ \btxjtitlefont
  {\btxifchangecase {{R}enormalization in quantum field theory and the
  {R}iemann-{H}ilbert problem {II}: the \(\beta\)-function, diffeomorphisms and
  the renormalization group}{{R}enormalization in quantum field theory and the
  {R}iemann-{H}ilbert problem {II}: the \(\beta\)-function, diffeomorphisms and
  the renormalization group}}.
\newblock \btxjournalfont {Commun. Math. Phys.}, 216:215--241, 2000.
\newblock arXiv:hep-th/0003188v1.

\bibitem {Kreimer_Anatomy}
\btxnamefont {\btxlastnamefont {{D. Kreimer}}}\btxauthorcolon\ \btxjtitlefont
  {\btxifchangecase {{A}natomy of a gauge theory}{{A}natomy of a gauge
  theory}}.
\newblock \btxjournalfont {Annals Phys.}, 321:2757--2781, 2006.
\newblock arXiv:hep-th/0509135v3.

\bibitem {vSuijlekom_QED}
\btxnamefont {\btxlastnamefont {{W. D. van Suijlekom}}}\btxauthorcolon\
  \btxjtitlefont {\btxifchangecase {{T}he {H}opf algebra of {F}eynman graphs in
  {QED}}{{T}he {H}opf algebra of {F}eynman graphs in {QED}}}.
\newblock \btxjournalfont {Lett. Math. Phys.}, 77:265--281, 2006.
\newblock arXiv:hep-th/0602126v2.

\bibitem {vSuijlekom_QCD}
\btxnamefont {\btxlastnamefont {{W. D. van Suijlekom}}}\btxauthorcolon\
  \btxjtitlefont {\btxifchangecase {{R}enormalization of gauge fields: {A}
  {H}opf algebra approach}{{R}enormalization of gauge fields: {A} {H}opf
  algebra approach}}.
\newblock \btxjournalfont {Commun. Math. Phys.}, 276:773--798, 2007.
\newblock arXiv:hep-th/0610137v1.

\bibitem {vSuijlekom_BV}
\btxnamefont {\btxlastnamefont {{W. D. van Suijlekom}}}\btxauthorcolon\
  \btxjtitlefont {\btxifchangecase {{T}he structure of renormalization {H}opf
  algebras for gauge theories {I}: {R}epresenting {F}eynman graphs on
  {BV}-algebras}{{T}he structure of renormalization {H}opf algebras for gauge
  theories {I}: {R}epresenting {F}eynman graphs on {BV}-algebras}}.
\newblock \btxjournalfont {Commun. Math. Phys.}, 290:291--319, 2009.
\newblock arXiv:0807.0999v2 [math-ph].

\bibitem {Kreimer_vSuijlekom}
\btxnamefont {\btxlastnamefont {{D. Kreimer}}} \btxandlong {}\ \btxnamefont
  {\btxlastnamefont {{W. D. van Suijlekom}}}\btxauthorcolon\ \btxjtitlefont
  {\btxifchangecase {{R}ecursive relations in the core {H}opf
  algebra}{{R}ecursive relations in the core {H}opf algebra}}.
\newblock \btxjournalfont {Nucl. Phys. B}, 820:682--693, 2009.
\newblock arXiv:0903.2849v1 [hep-th].

\bibitem {Kreimer_QG1}
\btxnamefont {\btxlastnamefont {{D. Kreimer}}}\btxauthorcolon\ \btxjtitlefont
  {\btxifchangecase {{A} remark on quantum gravity}{{A} remark on quantum
  gravity}}.
\newblock \btxjournalfont {Annals Phys.}, 323:49--60, 2008.
\newblock arXiv:0705.3897v1 [hep-th].

\bibitem {Kreimer_QG2}
\btxnamefont {\btxlastnamefont {{D. Kreimer}}}\btxauthorcolon\ \btxjtitlefont
  {\btxifchangecase {{N}ot so non-renormalizable gravity}{{N}ot so
  non-renormalizable gravity}}.
\newblock \btxjournalfont {"Quantum Field Theory: Competitive Models",
  B.Fauser, J.Tolksdorf, E.Zeidlers, eds., Birkhaeuser (2009)}, 2008.
\newblock arXiv:0805.4545v1 [hep-th].

\bibitem {Prinz_2}
\btxnamefont {\btxlastnamefont {{D. Prinz}}}\btxauthorcolon\ \btxjtitlefont
  {\btxifchangecase {{A}lgebraic {S}tructures in the {C}oupling of {G}ravity to
  {G}auge {T}heories}{{A}lgebraic {S}tructures in the {C}oupling of {G}ravity
  to {G}auge {T}heories}}.
\newblock \btxjournalfont {Annals Phys.}, 426:168395, 2021.
\newblock arXiv:1812.09919v3 [hep-th].

\bibitem {Prinz_3}
\btxnamefont {\btxlastnamefont {{D. Prinz}}}\btxauthorcolon\ \btxjtitlefont
  {\btxifchangecase {{G}auge {S}ymmetries and {R}enormalization}{{G}auge
  {S}ymmetries and {R}enormalization}}.
\newblock \btxjournalfont {Math. Phys. Anal. Geom.}, 25(3):20, 2022.
\newblock arXiv:2001.00104v4 [math-ph].

\bibitem {Prinz_4}
\btxnamefont {\btxlastnamefont {{D. Prinz}}}\btxauthorcolon\ \btxjtitlefont
  {\btxifchangecase {{G}ravity-{M}atter {F}eynman {R}ules for any
  {V}alence}{{G}ravity-{M}atter {F}eynman {R}ules for any {V}alence}}.
\newblock \btxjournalfont {Class. Quantum Grav.}, 38:215003, 2021.
\newblock arXiv:2004.09543v4 [hep-th].

\bibitem {Prinz_5}
\btxnamefont {\btxlastnamefont {{D. Prinz}}}\btxauthorcolon\ \btxtitlefont
  {\btxifchangecase {{T}he {BRST} {D}ouble {C}omplex for the {C}oupling of
  {G}ravity to {G}auge {T}heories}{{T}he {BRST} {D}ouble {C}omplex for the
  {C}oupling of {G}ravity to {G}auge {T}heories}}, 2022.
\newblock arXiv:2206.00780v1 [hep-th].

\bibitem {Prinz_6}
\btxnamefont {\btxlastnamefont {{D. Prinz}}}\btxauthorcolon\ \btxtitlefont
  {\btxifchangecase {{S}ymmetric {G}host {L}agrange {D}ensities for the
  {C}oupling of {G}ravity to {G}auge {T}heories}{{S}ymmetric {G}host {L}agrange
  {D}ensities for the {C}oupling of {G}ravity to {G}auge {T}heories}}, 2022.
\newblock arXiv:2207.07593v1 [hep-th].

\bibitem {Prinz_7}
\btxnamefont {\btxlastnamefont {{D. Prinz}}}\btxauthorcolon\ \btxtitlefont
  {\btxifchangecase {{T}ransversality in the {C}oupling of {G}ravity to {G}auge
  {T}heories}{{T}ransversality in the {C}oupling of {G}ravity to {G}auge
  {T}heories}}, 2022.
\newblock arXiv:2208.14166v1 [hep-th].

\bibitem {Prinz_PhD}
\btxnamefont {\btxlastnamefont {{D. Prinz}}}\btxauthorcolon\ \btxtitlefont
  {{R}enormalization of {G}auge {T}heories and {G}ravity}.
\newblock \btxphdthesis {}, Humboldt University of Berlin, 2022.
\newblock Available at \url{https://doi.org/10.18452/25401} and
  arXiv:2210.17510v1 [hep-th].

\bibitem {Ward}
\btxnamefont {\btxlastnamefont {{J. C. Ward}}}\btxauthorcolon\ \btxjtitlefont
  {\btxifchangecase {{A}n {I}dentity in {Q}uantum {E}lectrodynamics}{{A}n
  {I}dentity in {Q}uantum {E}lectrodynamics}}.
\newblock \btxjournalfont {Phys. Rev.}, 78:182, 1950.

\bibitem {Takahashi}
\btxnamefont {\btxlastnamefont {{Y. Takahashi}}}\btxauthorcolon\ \btxjtitlefont
  {\btxifchangecase {{O}n the {G}eneralized {W}ard {I}dentity}{{O}n the
  {G}eneralized {W}ard {I}dentity}}.
\newblock \btxjournalfont {Nuovo Cim.}, 6:371, 1957.

\bibitem {tHooft}
\btxnamefont {\btxlastnamefont {{G. 't Hooft}}}\btxauthorcolon\ \btxjtitlefont
  {\btxifchangecase {{R}enormalization of {M}assless {Y}ang-{M}ills
  {F}ields}{{R}enormalization of {M}assless {Y}ang-{M}ills {F}ields}}.
\newblock \btxjournalfont {Nucl. Phys. B}, 33 (1):173--199, 1971.

\bibitem {Taylor}
\btxnamefont {\btxlastnamefont {{J. C. Taylor}}}\btxauthorcolon\ \btxjtitlefont
  {\btxifchangecase {{W}ard identities and charge renormalization of the
  {Y}ang-{M}ills field}{{W}ard identities and charge renormalization of the
  {Y}ang-{M}ills field}}.
\newblock \btxjournalfont {Nucl. Phys. B}, 33 (2):436--444, 1971.

\bibitem {Slavnov}
\btxnamefont {\btxlastnamefont {{A. A. Slavnov}}}\btxauthorcolon\
  \btxjtitlefont {\btxifchangecase {{W}ard identities in gauge theories}{{W}ard
  identities in gauge theories}}.
\newblock \btxjournalfont {Theor. Math. Phys.}, 10 (2):99--104, 1972.

\bibitem {DeWitt_I}
\btxnamefont {\btxlastnamefont {{B. S. DeWitt}}}\btxauthorcolon\ \btxjtitlefont
  {\btxifchangecase {{Q}uantum {T}heory of {G}ravity. {I}. {T}he {C}anonical
  {T}heory}{{Q}uantum {T}heory of {G}ravity. {I}. {T}he {C}anonical {T}heory}}.
\newblock \btxjournalfont {Phys. Rev.}, 160:1113--1148, 1967.

\bibitem {Giulini_Kiefer}
\btxnamefont {\btxlastnamefont {{D. Giulini}}} \btxandlong {}\ \btxnamefont
  {\btxlastnamefont {{C. Kiefer}}}\btxauthorcolon\ \btxjtitlefont
  {\btxifchangecase {{W}heeler--{D}e{W}itt metric and the attractivity of
  gravity}{{W}heeler--{D}e{W}itt metric and the attractivity of gravity}}.
\newblock \btxjournalfont {Phys. Lett. A}, 193:21--24, 1994.
\newblock arXiv:gr-qc/9405040v2.

\end{thebibliography}
\bibliographystyle{babunsrt}

\end{document}